\documentclass[aps, pra, twocolumn, showpacs,superscriptaddress,10pt]{revtex4-1}

\usepackage{amsfonts,amssymb,amsmath}
\usepackage{graphics,graphicx,epsfig}
\usepackage{amsthm}
\usepackage{color}
\usepackage{epstopdf}

\newtheorem{theorem}{Theorem}

\newtheorem{lemma}{Lemma}
\newtheorem{remark}{Remark}
\newtheorem{propo}{Proposition}
\newcommand{\be}{\begin{eqnarray}}
\newcommand{\ee}{\end{eqnarray}}
\newcommand{\bpr}{\begin{propo}}
\newcommand{\epr}{\end{propo}}
\newcommand{\ble}{\begin{lemma}}
\newcommand{\ele}{\end{lemma}}
\newcommand{\bpf}{\begin{proof}}
\newcommand{\epf}{\end{proof}}
\newcommand{\ket}[1]{\left | #1 \right\rangle}
\newcommand{\vket}[1]{\left | #1 \right\rangle\rangle}
\newcommand{\bra}[1]{\left \langle #1 \right |}
\newcommand{\vbra}[1]{\langle\left \langle #1 \right |}

\newcommand{\Tr}{\mathrm{Tr}}
\newcommand{\G}{\mathcal{G}}
\newcommand{\Ham}{\mathcal{H}}
\newcommand{\rhoP}{\rho_{\Phi^+}}

\newcommand{\Gp}{G_{\pi}}
\newcommand{\rkp}{\rho_{\textrm{k-prod}}}

\newcommand{\Ll}{\mathcal L}
\newcommand{\ten}{\otimes}
\newcommand{\ts}{\mathcal T}
\newcommand{\non}{\nonumber\\}

\newcommand{\proj}[1]{\ket{#1}\bra{#1}}
\renewcommand{\epsilon}{\varepsilon}

\def\openone{\rm 1 \hspace{-0.9mm}l}

\def\matone{\mathbf 1}
\usepackage[T1]{fontenc}

\begin{document}

\title{Unified approach to geometric and positive-map-based non-linear entanglement identifiers}

\author{Marcin Markiewicz} \email{marcinm495@gmail.com}    
\affiliation{Institute of Physics, Jagiellonian University, \L{}ojasiewicza 11,
30-348 Krak\'ow, Poland} 

\author{Adrian Ko\l{}odziejski}
\affiliation{Institute of Theoretical Physics and Astrophysics, Faculty of Mathematics, Physics and Informatics, University of Gda\'nsk, 80-308 Gda\'nsk, Poland}

\author{Zbigniew Pucha{\l}a}
\affiliation{Institute  of  Theoretical  and  Applied  Informatics,
Polish  Academy  of  Sciences,  Ba{\l}tycka  5,  44-100  Gliwice,  Poland}
\affiliation{Institute of Physics, Jagiellonian University, \L{}ojasiewicza 11,
30-348 Krak\'ow, Poland} 

\author{Adam Rutkowski}
\affiliation{Institute of Theoretical Physics and Astrophysics, Faculty of Mathematics, Physics and Informatics, University of Gda\'nsk, 80-308 Gda\'nsk, Poland}

\author{Tomasz Tylec}
\affiliation{Institute of Theoretical Physics and Astrophysics, Faculty of Mathematics, Physics and Informatics, University of Gda\'nsk, 80-308 Gda\'nsk, Poland}

\author{Wies{\l}aw Laskowski}
\affiliation{Institute of Theoretical Physics and Astrophysics, Faculty of Mathematics, Physics and Informatics, University of Gda\'nsk, 80-308 Gda\'nsk, Poland}

\date{\today}


\begin{abstract}
Detecting  quantumness of correlations (especially entanglement) is a very hard task even in the simplest case i.e. two-partite quantum systems. Here we provide an analysis whether there exists a relation between two most popular types of  entanglement identifiers: the first one based on positive maps and not directly applicable in laboratory and the second one --- geometric entanglement identifier which is based on specific Hermiticity-preserving maps. We show a profound relation between those two types of entanglement criteria.   Hereunder we have proposed  a general  framework of nonlinear functional  entanglement identifiers  which allows us to construct new experimentally friendly entanglement criteria. 
\end{abstract}


\maketitle
\section{Introduction}
The problem of efficient characterization of quantum entanglement lies at the heart of quantum information science \cite{Guhne09, HHH09}. It is highly nontrivial even in the simplest case of a bipartite scenario, since whenever the dimension of the entire state space is greater than $6$ there are no universal efficient methods to check, whether a given state is entangled or not \cite{HHH96}. When it comes to a general characterization of bipartite entanglement, there were proposed two seemingly different complete  (necessary and sufficient) entanglement criteria. The first one is the positive map criterion \cite{HHH96}:
\be
\rho \textrm{ is entangled } \iff \exists_{\Lambda} (\openone\otimes \Lambda)\rho \ngeq 0.
\label{mapmain}
\ee
Intuitively this condition is based on the fact, that performing a locally allowed operation on a part of an entangled system may spoil the structure of the entire state, making it unphysical. The most well known map $\Lambda$ is the transposition map $\mathcal T$, which detects entanglement of all two-qubit and qubit-qutrit states \cite{HHH96}.
There has been an attempt to generalize this approach for detection of genuine multipartite entanglement \cite{HHH01, Huber14, Wolk14,  Clivas16}, however no general construction of such criteria, which would be optimal for arbitrary state exists.

The second group of entanglement indicators involves so called geometric criteria \cite{LMPZ11,Laskowski12, Markiewicz13, Laskowski13, Laskowski15}, which in the simplest bipartite case read \cite{BBLPZ08}:
\be
&&\rho \textrm{ is entangled } \iff \nonumber\\
&&\exists_{G|\Tr(\rho G[\rho])\geq 0} \forall_{\rho_1,\rho_2} \Tr(\rho.G[\rho_1\otimes\rho_2])<\Tr(\rho.G[\rho]).
\label{tensormain1}
\ee
These criteria are based on specific Hermiticity-preserving maps $G$, which fulfill a specific positivity  condition $\Tr(\rho G[\rho])\geq 0$, and can be treated as metric tensors.
The above condition is based on approximating the norm of investigated state by a scalar product with any product state in any convenient metric, and in fact it can be treated as the most general implementation of the Hahn-Banach separation theorem.
This group of criteria can be directly generalized to provide a necessary and sufficient condition of arbitrary level of separability (involving genuine multipartite entanglement) for any number $N$ of finite-dimensional systems \cite{LMPZ11}:
\be
&&\rho \textrm{ is not $k$-separable } \iff \nonumber\\
&&\forall_{\pi}\exists_{\Gp|\Tr(\rho \Gp[\rho])\geq 0} \forall_{\rkp} 
\Tr(\rho.\Gp[\rkp])<\Tr(\rho.\Gp[\rho]),\non
\label{tensormulti}
\ee
in which $\pi$ denotes a partition of $N$ systems into $k$ subsystems, $\Gp$ is a metric operator consistent with the partition $\pi$, and $\rkp$ denotes a state which is $k$-product with respect to the partition $\pi$. If we take $k=2$, the condition rejects the possibility that a state $\rho$ is biseparable in some partition, therefore proving that it is genuinely $N$-partite entangled.
The geometric entanglement criteria belong to a wider class of geometric nonlinear entanglement indicators (often called nonlinear entanglement witnesses) \cite{Guhne07, Moroder08, Arrazola12, obwiednie}, however they are distinguished by the fact of completeness.

The question arises whether the above two families of complete bipartite entanglement criteria are in any way related to each other, in a deeper sense from the fact that they define the same sets of states. In this work we explicitly show that such a deep relation exists, namely each entanglement identifier of the form \eqref{tensormain1} can be translated to a positive-map-based criterion \eqref{mapmain}, which detects entanglement of at least the same set of states. In order to obtain the relation, we introduce a much more general framework for entanglement detection based on nonlinear functionals, which itself has no geometric interpretation. Within this scenario we show, that each functional-based nonlinear entanglement identifier can be effectively translated to a positive-map-based criterion, and that the geometrical criteria \eqref{tensormain1} are a subset of such identifiers. The construction of a positive map which corresponds to a functional identifier is effectively one-way, which means that it cannot be effectively reversed. Namely, the construction of a map given nonlinear identifier is simple and straightforward, however finding identifiers which would correspond to a fixed positive map involves solving highly nonlinear systems of equations and therefore cannot be done efficiently in general case. Interestingly, it can be easily done in the case of a positive partial transpose (PPT) map for two-qubit systems, which is distinguished as a universal map detecting enetanglement of arbitrary two-qubits states \cite{Peres96, HHH96}. Further we show that our generalized framework of nonlinear functional entanglement identifiers allows for an easy construction of new experimentally friendly entanglement criteria which are much more efficient than the geometric ones. As a final remark we discuss why the relation between functional criteria and map-based criteria cannot be extended for the case of multipartite entanglement.

\section{General non-linear bipartite entanglement identifier and its positive-map-based version}
\subsection{Functional form of the identifier}
In this section we define within a few steps the most general form of a non-linear entanglement identifier and discuss its properties. Let us assume we deal with a bipartite system of local dimensions $d_A$ and $d_B$ respectively. The starting point of the construction is a (possibly non-linear) map $\G$ acting on the space of matrices of dimension $d_A\cdot d_B$. The only restriction on the map $\G$ is that it must be Hermiticity-preserving. Further let us define a non-linear $2$-form ($2$-argument functional):
\be
\tilde \omega_{\G}(\sigma,\rho) = \Tr(\sigma\,\G[\rho]),
\label{TildeOmega}
\ee
which acts on arbitrary states  on the joint space $\Ham_A\otimes\Ham_B$. Due to the assumption of Hermiticity-preserving of $\G$, the functional $\tilde\omega_{\G}$ takes only real values. Therefore it is meaningful to define the following functional:
\be
\tilde\omega_0(\rho)=\max_{\sigma_{\textrm{sep}}} \tilde \omega_{\G}(\sigma_{\textrm{sep}},\rho),
\label{OmegaMax}
\ee
which intuitively gives the maximal overlap between a given state $\rho$ and a set of separable states. Since $\tilde\omega_{\G}$ is linear in the first argument, the maximization in $\tilde\omega_0$ can be done only over a set of pure product states:
\be
\max_{\sigma_{\textrm{sep}}} \tilde \omega_{\G}(\sigma_{\textrm{sep}},\rho)&=&\max_{\{p_i\},\sigma_{\textrm{prod}}}\sum_i p_i \tilde\omega_{\G}(\sigma_{\textrm{prod}},\rho)\nonumber\\
&=&
 \max_{\sigma_{\textrm{prod}}}\tilde\omega_{\G}(\sigma_{\textrm{prod}},\rho).
\label{OmegaProd}
\ee
Now we introduce the following functional:
\be
\omega_{\G}(\rho)=\tilde\omega_0(\rho)-\tilde\omega_{\G}(\rho,\rho),
\label{DefOmega}
\ee
which turns out to be a general entanglement identifier:
\begin{propo}
For any Hermiticity-preserving map $\G$, the functional $\omega_{\G}(\rho)$ is an entanglement identifier, namely:
\be
\omega_{\G}(\rho)<0\Longrightarrow \rho \textrm{ is entangled. }
\label{MainCond}
\ee
\label{PropMainCond}
\end{propo}
\begin{proof}
Indeed, $\omega_{\G}(\rho)<0$ implies that:
\be
\max_{\sigma_{\textrm{prod}}}\tilde\omega_{\G}(\sigma_{\textrm{prod}},\rho)<\tilde\omega_{\G}(\rho,\rho),
\ee
which cannot hold for any separable state $\rho$.
\end{proof}
\begin{remark}
Note that the condition \eqref{MainCond} is purely algebraic, and in principle need not have any geometric interpretation. 
\end{remark}
\begin{remark}
Note that the functional $\omega_{\G}$ is nonlinear even for a linear map $\G$ since the maximization present in \eqref{DefOmega} is a highly nonlinear operation.
\end{remark}

\subsection{Positive-map-based condition corresponding to an identifier}
In full analogy to the case of ordinary entanglement witnesses \cite{Guhne03, Wolk14} each entanglement identifier of the form \eqref{MainCond} gives rise to a positive map, which detects entanglement of at least the same set of states as the identifier. The difference is that in our construction the map explicitly depends on the state on which the functional identifier acts. To see this let us rewrite the condition \eqref{DefOmega} in a witness-like form:
\be
\omega_{\G}(\rho)=\tilde\omega_0(\rho)-\Tr(\rho\G[\rho])=\Tr\left(\rho(\tilde\omega_0(\rho)\matone-\G[\rho])\right),\nonumber\\
\label{DefOmega1}
\ee
and let us introduce new Hermiticity-preserving map $W_{\G}[\rho]=\tilde\omega_0(\rho)\matone-\G[\rho]$. Since now the condition \eqref{MainCond}
takes the following form:
\be
\Tr(\rho W_{\G}[\rho])<0,
\label{MainCondW}
\ee
the operator $W_{\G}[\rho]$ (seen as a state $\rho$ transformed by the map $W_{\G}$) can be seen as a nonlinear state-dependent entanglement witness. Now the standard relation between entanglement witnesses and maps can be used in the following way \cite{HHH96}. Let $\{\sigma_i^A\}$ and $\{\sigma_i^B\}$ be two Hermitian orthonormal operator bases with normalization given by $\Tr(\sigma_i\sigma_j)=2\delta_{ij}$. Let us denote by $w^{\G,\rho}_{ij}$ matrix elements of $W_{\G}[\rho]$ in the introduced bases:
\be
w^{\G,\rho}_{ij}=\Tr(W_{\G}[\rho] \sigma_i^A\otimes\sigma_j^B).
\label{WCoef}
\ee
Using these coefficients we can represent the operator $W_{\G}[\rho]$ in the following form:
\be
W_{\G}[\rho]=\frac{1}{4}\sum_{ij}w^{\G,\rho}_{ij}\sigma_i^A\otimes\sigma_j^B,
\ee
in which the coefficient $\tfrac{1}{4}$ arises as a compensation for the nontrivial basis normalization $\Tr(\sigma_i\sigma_j)=2\delta_{ij}$.
We will use the following simple version of the celebrated Choi-Jamio\l{}kowski isomorphism \cite{Jamiolkowski72, Choi75, Choi82} (for an intuitive introduction to the subject see Appendix \ref{ApChoi}):
\ble
For any local Hermitian operator bases $\{\sigma_i^A\}$ and $\{\sigma_i^B\}$ 
we define  a linear map $\Lambda_{\G,\rho}$ as
\be
\Lambda_{\G,\rho}[\lambda]=\sum_{i,j}\frac{1}{4}w^{\G,\rho}_{ij}\,\Tr\left(\sigma^A_i\,\lambda\,\right)\sigma_j^B,
\label{Lmap}
\ee
then the following identity holds:
\begin{equation}\label{MainIso}
W_{\G}[\rho]=
\left(
\openone \otimes \Lambda_{\G,\rho}
\right)
\!\left[\frac{1}{2}\sum_{m}\sigma^A_m\ten \sigma^A_m\right].
\end{equation}
\ele
\begin{proof}
\be
&&
\left(\openone \otimes \Lambda_{\G,\rho}\right)
\!\left[\frac{1}{2}\sum_{m}\sigma^A_m\ten \sigma^A_m\right] \\
&&=\frac{1}{2}\sum_m\left(\sigma_m^A\ten\sum_{i,j}\frac{1}{4}w^{\G,\rho}_{ij}\,\Tr\left(\sigma^A_i\,\sigma_m\,\right)\sigma_j^B\right)\non
&&=\frac{1}{2}\sum_m\left(\sigma_m^A\ten\sum_{i,j}2\delta_{im}\frac{1}{4}w^{\G,\rho}_{ij}\sigma_j^B\right)\non
&&=\sum_{ij}\left(\frac{1}{4}w^{\G,\rho}_{ij}\sigma_i^A\ten\sigma_j^B\right)=W_{\G}[\rho].
\ee
\end{proof}
Note that the map $\Lambda_{\G,\rho}$ is positive, which holds due to the Choi's theorem. Choi's theorem \cite{Choi75} states, that a map $\Lambda$ is positive if and only if its matrix representation $\rho_{\Lambda}$ given by  $\rho_{\Lambda}=(\openone\ten\Lambda)\rho_{\Phi^+}$ is block-positive, in which the operator $\rho_{\Phi^+}$ is given by:
\be
\rho_{\Phi^+}=\sum_m \sigma^A_m\ten\left(\sigma^A_m\right)^T.
\label{rofi}
\ee
This means that:
$\Tr(\rho_1\ten\rho_2 \rho_{\Lambda})\geq0$ for any pure states $\rho_1$ and $\rho_2$. The operator $W_{\G}[\rho]$ is block-positive for any $G$ and $\rho$, since:
\be
&&\Tr(\rho_1\ten\rho_2 W_{\G}[\rho])=\Tr(\rho_1\ten\rho_2 (\tilde\omega_0(\rho)\openone-\G[\rho]))\non
&&=\tilde\omega_0(\rho)-\Tr(\rho_1\ten\rho_2 \G[\rho])\geq 0.
\ee
The last inequality holds due to the definition of $\tilde\omega_0(\rho)$ \eqref{OmegaMax}. This implies that the map $\Lambda_{G,\rho}$ is a positive map for any $\G$ and $\rho$.

Before proceeding further we would need to introduce three further technical tools. 
Firstly, we need a notion of a dual map, which allows one to switch the action 
of the map inbetween two operators under trace. We call a map
$\Lambda^{\textrm d}\in\Ll(\Ll(\Ham_B),\Ll(\Ham_A))$ \emph{dual} to the map 
$\Lambda\in\Ll(\Ll(\Ham_A), \Ll(\Ham_B))$ 
if is is dual with respect to Hilbert Schmidt inner product, i.e.:
\be \!\!\!\!\!\!\!
\forall_{\alpha_1\in\Ll(\Ham_B),\alpha_2\in  \Ll(\Ham_A)} 
\Tr(\alpha_1^* \Lambda[\alpha_2])=\Tr((\Lambda^{\textrm d}[\alpha_1])^*\alpha_2).\non
\label{rdual}
\ee
The proof of the following proposition is postponed to the Appendix 
\ref{proofLdual}.
\begin{propo}
A map dual to \eqref{Lmap} is given by the following formula:
\be
\Lambda_{\G,\rho}^d[\lambda]=\sum_{ij} w_{ij}^{\G,\rho}\,\Tr\left(\sigma^B_j\lambda\right)\sigma^{A}_i.
\label{Ldual}
\ee
\label{propLdual}
\end{propo}
Secondly we would need a property connected with taking the dual of a map 
composed with the transposition map $\mathcal T$. Since 
$(\Lambda_1\circ\Lambda_2)^d=\Lambda_2^d\circ\Lambda_1^d$, we obtain the 
following fact:
\begin{propo}
For any linear map $\Lambda$ the following property holds
\be
(\Lambda\circ\ts)^d=\ts\circ\Lambda^{\textrm d}.
\ee
\label{propTdual}
\end{propo}
Finally we need a characterization of the \emph{maximally entangled operator} $\sum_m\sigma^A_m\otimes\sigma^A_m$ that appears in the isomorphism \eqref{MainIso}:
\begin{propo}
Let $\{\sigma_i\}$ be a collection of Hermitian matrices which forms an 
orthonormal basis on the Hilbert-Schmidt space of $d\times d$ matrices. Then 
the following relation holds:
\be
\sum_{m}\sigma_m\ten \sigma_m=(\openone\ten \ts)\rho_{\Phi^+},
\label{fiplusone}
\ee
in which $\rho_{\Phi^+}$ \eqref{rofi} is a projector onto an unnormalized \emph{maximally entangled state} $\ket{\Phi^+}=\sum_i\ket{ii}$ in dimension $d$.
\label{propFiPlus}
\end{propo}
\begin{proof}
See Appendix \ref{proofFiPlus}.
\end{proof}
Now we are ready to present a transition from a functional entanglement identifier \eqref{MainCondW} to a positive-map condition in a concise way. Firstly using the introduced notation we can rewrite \eqref{MainIso} in the following way:
\be
W_{\G}[\rho]=\left(\openone\otimes\Lambda_{\G,\rho}\circ\mathcal T\right)[\rhoP],
\label{W1}
\ee
in which we put the factor of $\tfrac{1}{2}$ into the unnormalized state  $\rhoP$. Then the entanglement identifier \eqref{MainCondW} reads:
\be
\Tr\left(\rho \left(\openone\otimes\Lambda_{\G,\rho}\circ\mathcal T\right)[\rhoP]\right)<0,
\label{W2}
\ee
and after taking the dual maps we obtain:
\be
\Tr\left(\rhoP \left(\openone\otimes\mathcal T\circ\Lambda^{\textrm d}_{\G,\rho}\right)[\rho]\right)<0.
\label{W3}
\ee
Since $\rhoP$ is a one-dimensional projector, and the map $\mathcal T\circ\Lambda^{\textrm d}_{\G,\rho}$ is positive, the condition \eqref{W2} \emph{implies} that:
\be
\left(\openone\otimes\mathcal T\circ\Lambda^{\textrm d}_{\G,\rho}\right)[\rho]\ngeq 0.
\ee
All the discussion can be summarized as follows:
\begin{theorem}
If an entanglement identifier $\omega_{\G}[\rho]$ \eqref{DefOmega} detects bipartite entanglement of a state $\rho$ for some Hermiticity-preserving map $\G$, then a positive map $\mathcal T\circ \Lambda^{\textrm d}_{\G,\rho}$ also detects entanglement of $\rho$ via condition:
\be
\left(\openone\otimes\mathcal T\circ\Lambda^{\textrm d}_{\G,\rho}\right)[\rho]\ngeq 0.
\label{MainCondMap}
\ee
Given orthonormal Hermitian bases $\{\sigma^A_i\}$ and $\{\sigma^B_i\}$ the map $\Lambda^{\textrm d}_{\G,\rho}$ can be explicitly expressed as:
\be
\Lambda_{\G,\rho}^d[\lambda]=\sum_{ij} \frac{1}{4} w_{ij}^{\G,\rho}\,\Tr\left(\sigma^B_j\lambda\right)\sigma^{A}_i,
\label{MainMap}
\ee
in which $w_{ij}^{\G,\rho}$ are matrix elements of an operator $\tilde\omega_0(\rho)\matone-\G[\rho]$.
\end{theorem}
Two comments are necessary here. Firstly the positive map $\mathcal T\circ \Lambda^{\textrm d}_{\G,\rho}$ can in principle detect entanglement of a broaden class of states than the functional $\omega_{\G}$, which means that there can exist a state $\lambda$ such that $\omega_{\G}[\lambda]\geq 0$, but $\left(\openone\otimes\mathcal T\circ\Lambda^{\textrm d}_{\G,\rho}\right)[\lambda]\ngeq 0$. We will give an instructive example in the following sections. Secondly although the map $\Lambda^{\textrm d}_{\G,\rho}$ is linear, the entanglement condition \eqref{MainCondMap} is not linear, since for a given fixed $\G$ the map $\Lambda^{\textrm d}_{\G,\rho}$ is explicitly adjusted to the state it acts on.

\subsection{Experimentally friendly form of an entanglement identifier}
We call an entanglement identifier \emph{experimentally friendly} if it can be efficiently evaluated using only local measurements on an investigated state. 
Any bipartite quantum state $\rho$ can be expressed in local Hermitian bases $\{\sigma^A_i\}$ and $\{\sigma^B_j\}$ fulfilling $\Tr(\sigma^A_i\sigma^A_j)=\Tr(\sigma^B_i\sigma^B_j)=2\delta_{ij}$ in the following form:
\be
\rho=\frac{1}{4}\sum_{ij} T_{ij}\sigma^A_i\ten\sigma^B_j,
\label{ctensor}
\ee 
in which the normalization is chosen such that the coefficients $T_{ij}$ fulfill $-1\geq T_{ij}\geq 1$ whether at least one index is non zero. $T_{00}=\tfrac{2}{d}$ is chosen such that $\Tr(\rho)=1$ and the prefactor of $4$ is chosen to compensate for the factor $2$ in the normalization of the bases.
The set of coefficients $\{T_{ij}\}$ is called a \emph{correlation tensor} of a state $\rho$, since it represents average values of correlations:
\be
T_{ij}=\Tr(\rho \sigma^A_i\otimes\sigma^B_j),
\ee
and transforms as a tensor under local unitary operations done on subsystems. The entanglement identifier $\omega_{\G}$ is experimentally friendly if and only if $\G[\rho]$ can be effectively expressed as some function of elements of a correlation tensor $T_{ij}$ of $\rho$. This holds for example whenever the map $\G$ is linear, since then:
\be
\label{GLin}
\G[\rho]=\frac{1}{4}\sum_{ij} T_{ij}\G(\sigma^A_i\ten\sigma^B_j).\non
\ee
However, any linear transformation of a simple tensor $\sigma^A_i\ten\sigma^B_j$ can be expressed in terms of a tensor operator $\hat G$, which has the following form in chosen coordinates:
\be
\G(\sigma^A_i\ten\sigma^B_j)=\sum_{kl}G^{ij}_{kl}\sigma^A_k\ten\sigma^B_l.
\label{GLin1}
\ee
Joining \eqref{GLin} and \eqref{GLin1} we obtain:
\be
\label{GLinFin}
\G[\rho]&=&\frac{1}{4}\sum_{ijkl} T_{ij}G^{ij}_{kl}\sigma^A_k\ten\sigma^B_l\non
&=&\frac{1}{4}\sum_{kl} \left(\sum_{ij}G^{ij}_{kl}T_{ij}\right)\sigma^A_k\ten\sigma^B_l.
\ee
In the above formula we can think of a tensor operator $\hat G$ as acting on the correlation tensor $\hat T$ of the state $\rho$. This is because our operator basis $\sigma^A_k\ten\sigma^B_l$ is orthonormal and can be thought of as a \emph{Cartesian} basis in the space of operators on the tensor product of two Hilbert spaces. Therefore we do not need to distinguish between covariant and contravariant coordinates, and both the coefficients of a tensor and the basis elements transform in the same way.

If the map $\G$ is not linear the operator $\G[\rho]$ cannot be in general presented as a function of the coefficients $T_{ij}$. 
On the other hand we can construct a nonlinear map $\G$ which leads to experimentally friendly entanglement indentifier. We propose the following form of such a map:
\be
\label{GNonLin}
\G[\rho]=\frac{1}{4}\sum_{kl} \left(\sum_{ij}G^{ij}_{kl}f(T_{ij})\right)\sigma^A_k\ten\sigma^B_l,
\ee
in which $f(x)$ is an arbitrary real function applied elementwise to the tensor $\hat T$. We will show that such a construction leads to an extremely useful experimentally friendly entanglement identifiers.

\section{Examples}

\subsection{Linear geometrical entanglement criteria}

In this section we will use the introduced formalism for a reconsideration of the so called \emph{geometrical criteria for entanglement} introduced in \cite{BBLPZ08} and further developed in a series of works \cite{LMPZ11,Laskowski12, Markiewicz13, Laskowski13, Laskowski15}. All the criteria of this type are based on an entanglement identifier \eqref{DefOmega} generated by the map $\G$ which is linear \eqref{GLinFin} and which fulfills the following condition:
\be
\label{GPos}
\Tr(\rho\G[\rho])\geq 0,
\ee 
which should hold for all states $\rho$.
In the case when condition~\eqref{GPos} is satisfied for all Hermitian matrices,
it allows one to treat the above expression as a seminorm 
$||\rho||_{\G}=\Tr(\rho\G[\rho])$, which induces a \emph{pseudometric}:
\be
d_{\G}(\rho, \sigma)=||\rho-\sigma||_{\G}.
\label{GMetric}
\ee
The prefixes \emph{semi} and \emph{pseudo} denote the fact, that in these type of criteria one do not impose the condition:
\be
||\rho||_{\G}=0 \Longrightarrow \rho=\textbf{0},
\label{seminorm}
\ee
which means that a non-zero matrix can have zero norm. The metric 
\eqref{GMetric} generated by the linear map $\G$ \eqref{GPos} can be seen as a 
generalization of the Hilbert-Schmidt metric, whereas the 2-form 
$\tilde\omega_{\G}$ \eqref{TildeOmega} can be treated as an inner product. 
When the condition~\eqref{GPos} is satisfied for all Hermitian matrices, then 
$\G$ must be Hermiticity preserving and the form:
\begin{equation}
\Tr (X^* \G[X]) 
\end{equation}
is positive for any matrix $X$. This follows from linearity and decomposition 
of any matrix $X$ to the sum of Hermitian and anti-Hermitian matrix.
The above considerations lead to the conclusion, that $\G$ treated as a linear 
operator must be positive semidefinite.
The map $\G$ \eqref{GPos} itself in its tensor representation $\hat G$ 
\eqref{GLinFin} can be treated as a \emph{metric tensor}, however taken in the 
framework of multilinear algebra and not differential geometry. This is because 
in our construction the concrete choice of the linear map $\G$, which is 
positive in the sense of \eqref{GPos}, determines the global distance between 
any two vectors \eqref{GMetric}, whereas in differential geometry one starts 
from a metric tensor which determines an infinitesimal distance and this 
distance has to be extended to a global distance via integration.

The entanglement identifier \eqref{DefOmega} based on map \eqref{GPos} has a simple geometric interpretation. Namely, the entanglement condition \eqref{MainCond} can be explicitly expressed as:
\be
\max_{\sigma_{\textrm{prod}}}\tilde\omega_{\G}(\sigma_{prod},\rho)<\tilde\omega_{\G}(\rho,\rho),
\label{MainCondGLin}
\ee
which means that if some state $\rho$, treated here as a vector, follows the property that its inner product with any of the extremal elements of some convex set is strictly lower than the inner product of $\rho$ with itself, then the vector $\rho$ cannot be an element of the convex set. The convex set under consideration is typically chosen as the set of separable states, however it can be as well chosen in a different way, as for example the set of PPT states \cite{PPTStates} (which are all states with positive partial transpose).

Although the entanglement condition \eqref{MainCondGLin} is based on a linear map $\G$, and it can be presented in a witness-like form \eqref{MainCondW} it is in fact manifestly nonlinear due to two reasons. Firstly the RHS of \eqref{MainCondGLin} is for any linear $\G$ a quadratic form of matrix elements of the state $\rho$, secondly, the LHS of \eqref{MainCondGLin} is a maximization of the inner product with respect to extremal elements of some convex set, which is in principle a highly non-linear operation. Former works on metric entanglement criteria lacked direct visualization of this fact, therefore we provide it here. 
Let us assume that $\G[\rho]$ is determined via formula \eqref{GLinFin} by the so called \emph{standard improper} metric tensor $\hat G$, with elements defined by:
\be
G^{ij}_{kl}=\delta_{ik}\delta_{jl}[i,j\neq 0],
\label{StdMet}
\ee
in which $[p]$ denotes the logical value of proposition $p$. In simple words the metric tensor \eqref{StdMet} acting on the state determined by a correlation tensor $\hat T$ cancels all the local averages and keeps unchanged full two-point correlations. We consider two classes of states. The first is the class of two-parameter Bell-diagonal states of two qubits of the form:
\be
\rho_{BD}=a\proj{\Phi^+}+b\proj{\Phi^-}+\frac{(1-a-b)}{4}\openone\non,
\label{BellStates}
\ee
whereas the second class is the 3-parameter class of two-qubit states, the correlation tensor of which contains as nonzero terms only the diagonal full correlation terms $T_{11}=p$, $T_{22}=q$ and $T_{33}=r$. The second class of states, known as Weyl states, was thoroughly discussed in \cite{HH96}, whereas the first class is actually the subclass of it, which we discuss for the sake of simple geometrical visualization (see also \cite{Moroder08}). Let us first discuss the case of Bell diagonal states \eqref{BellStates}. This family of states has two parameters, which we represent as two orthogonal directions on the plane (see Fig. \ref{BDFig}). The set of all physical (properly normalized states), fulfilling $a+b\leq 1$ is represented by the right angle triangle, whereas the set of separable states from this class (which can be verified with PPT condition) is represented by the internal deltoid. The entanglement identifier \eqref{MainCondGLin} with standard metric tensor \eqref{StdMet} is in this case given by a quadratic condition $a+b+2ab-3(a^2+b^2)<0$, which geometrically gives an ellipse which circumscribes the deltoid of separable states. The condition \eqref{MainCondGLin} detects entanglement of states that lie outside the ellipse.
Let us now discuss the second family of states. The three  diagonal elements of the correlation tensor, which parametrize the family, can be represented as three orthogonal directions in space (see Fig. \ref{WSFig}). The set of all physical states is represented by the shaded tetrahedron, whereas the set of separable states is depicted by the inscribed octahedron. In this case the entanglement identifier \eqref{MainCondGLin} with the standard metric evaluates to a more nonlinear expression of the form:
\be
p^2+q^2+r^2-\max(|p|,|q|,|r|)<0,
\label{WStatesCond}
\ee
which is represented by the orange curved surface. The states which lie outside this surface are recognized as entanglement by the indicator.

\begin{figure}
\includegraphics[width=\columnwidth]{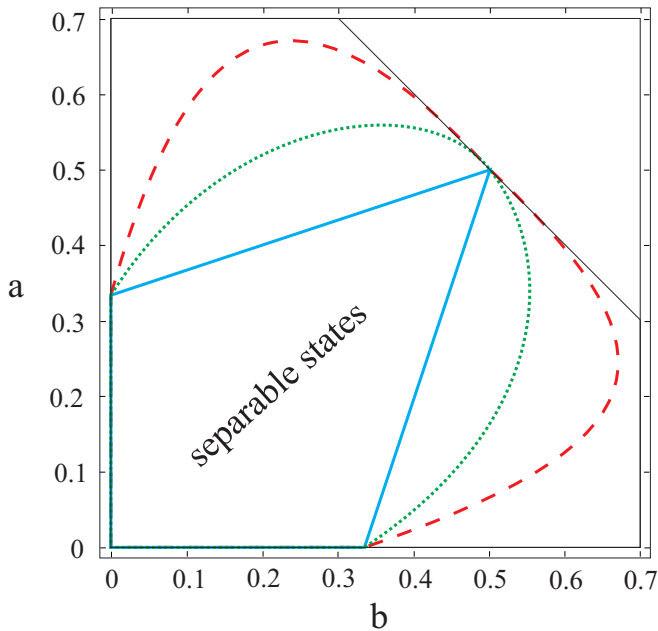}
\caption{Entanglement detection of Bell diagonal states \eqref{BellStates}. Separable states are enclosed in the central deltoid (blue solid line). The nonlinear indicator $\omega_{G}$ \eqref{MainCond} with standard metric \eqref{StdMet} detects entanglement of states outside the elipse (green dotted line). For comparison we present also the set of entangled states detected by the indicator \eqref{MainCond} generated by explicitly nonlinear map $\G$ \eqref{GNonLin} with $f(x)=x^3$ (red dashed line).}
\label{BDFig}
\end{figure}

\begin{figure}
\includegraphics[width=\columnwidth]{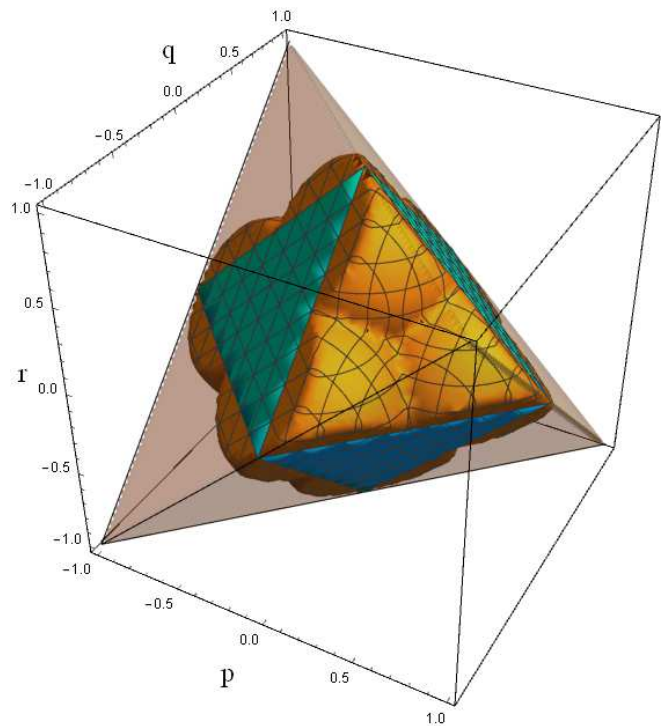}
\caption{Entanglement detection of Weyl states with standard metric $G$ \eqref{StdMet}. All physical states are enclosed in the grey tetrahedron. Separable states are enclosed in the central green octahedron. The nonlinear indicator $\omega_{G}$ \eqref{MainCond} detects entanglement of states outside the curved orange region.}
\label{WSFig}
\end{figure}

Insofar we discussed only one choice of the metric tensor $\hat G$, hence the question arises how one should search for a proper metric tensor to detect entanglement of a given class of states. As we saw in the examples, the standard metric tensor \eqref{StdMet} is not able to detect entanglement of all states from discussed families. Can we effectively find a metric tensor which detects entanglement of arbitrary bipartite state? The answer is negative. In \cite{BBLPZ08, LMPZ11} it was shown that for every state $\rho$ there exists a metric tensor $\hat G_{\textrm{opt}}$ that leads to detection of entanglement of this state, given by the following expression:
\be
\hat G_{\textrm{opt}}=|\rho-\rho_0)(\rho-\rho_0|,
\ee
in which $\rho_0$ denotes the closest separable state \emph{in the Hilbert-Schmidt metric} and the operator on the RHS is a projector in the Hilbert-Schmidt space of operators. Since the problem of finding the closest separable state to a given state is in general as complex as identifying whether the given state is separable or not, this choice of metric cannot be treated as effectively findable. In this point of the discussion we reach the main advantage of the general entanglement identifiers \eqref{MainCond}. We may try to guess the form of a map $\G$ and check whether it leads to a detection of entanglement of a given state. 

As we discussed in the previous section, every functional entanglement identifier of the form \eqref{MainCond} can be used to construct a positive map \eqref{MainMap}, which detects entanglement of at least the same states as the identifier. 
\begin{propo}
In the case of a linear map $\G$ \eqref{GLinFin}, the positive map corresponding to the identifier generated by $\G$, which is applied to a state $\rho$, is explicitly given by the formula:
\be
\Lambda_{\hat G,\rho}^d[\lambda]=\sum_{ijkl} \frac{1}{4}(\tilde{T}_{ij}-G^{kl}_{ij}T_{kl})\Tr\left(\sigma^B_j\lambda\right)\sigma^{A}_i,
\label{MapGLin}
\ee
in which the tensor $\tilde{T}$ is defined by:
\be
\tilde{T}_{ij}=
\begin{cases}
\frac{\sqrt{d_Ad_B}}{2}T_{\max}[\rho]\\
=\frac{\sqrt{d_Ad_B}}{2}\max_{\hat X,\hat Y} \sum_{klmn}X_kY_lG^{kl}_{mn}T_{mn}, \,\,i=j=0,\\
0 \textrm{ otherwise.}
\end{cases}\non
\label{TMaxDef}
\ee
\label{PropGTMax}
\end{propo}
\begin{proof}
See Appendix \ref{ProofGTMax}.
\end{proof}
Let us discuss two simple examples of such maps in the case of two-qubit systems taken for a standard metric tensor \eqref{StdMet} and two different initial states: the singlet state and the Werner state. 
The non-vanishing correlation tensor components of the Bell's singlet state ($| \phi^{-} \rangle = \frac{1}{\sqrt{2}} (|01 \rangle - |10 \rangle)$) are: $T_{xx} = T_{yy} = T_{zz} = - 1$. Hence, the map \eqref{MapGLin} takes the following form:
\be
\Lambda^{\textrm d}_{\hat G,|\phi^{-} \rangle}[\rho]&=& \frac{T_{max}[\rho]}{4} \Tr{(\rho)}\matone + \frac{1}{4} \Tr{(\sigma_x x)}\sigma_x \non
&&- \frac{1}{4} \Tr{(\sigma_y x)}\sigma_y + \frac{1}{4} \Tr{(\sigma_z x)}\sigma_z,
\label{SingletMap}
\ee
in which the function $T_{\max}$ is defined in \eqref{TMaxDef}. Applying the transformation \eqref{MainCondMap} with the above positive map to a singlet state, one gets an operator with eigenvalues $\lbrace -\frac{1}{4}, \frac{1}{4}, \frac{1}{4}, \frac{1}{4} \rbrace$. Therefore the map \eqref{SingletMap} detects entanglement of the singlet state. On the other hand it also detects the entanglement of the Werner state:
\be
\rho_W = v | \phi^{+} \rangle \langle \phi^{+} | + \frac{1-v}{4} \matone, 
\label{WerState}
\ee
in which $| \phi^{+} \rangle = \frac{1}{\sqrt{2}}(| 00 \rangle + | 11 \rangle)$.
Indeed, the eigenvalues of the operator $\left(\openone\ten\mathcal T\circ\Lambda^{\textrm d}_{\hat G,|\phi^{-} \rangle}\right)[\rho_W]$  are $\left\{\frac{1}{8} (1-3 v),\frac{v+1}{8},\frac{v+1}{8},\frac{v+1}{8}\right\}$, hence the entanglement of $\rho_W$ is detected for $v>\tfrac{1}{3}$. Now let us construct the map \eqref{MainMap} starting from the Werner's state \eqref{WerState}. The non-vanishing correlation tensor components of the Werner's state are: $T_{xx} = T_{zz} = - T_{yy} = v$. Hence, the map \eqref{MainMap} takes the following form:
\be
\Lambda^{\textrm d}_{\hat G, \rho_{W}}[\rho] &=& \frac{T_{max}[\rho]}{4} \Tr{(\rho)}\matone - \frac{v}{4} \Tr{(\sigma_x x)}\sigma_x \non
&&- \frac{v}{4} \Tr{(\sigma_y x)}\sigma_y - \frac{v}{4} \Tr{(\sigma_z x)}\sigma_z
\label{WernerMap}
\ee
Applying the transformation \eqref{MainCondMap} to the Werner state \eqref{WerState} one obtains operator with eigenvalues \newline $ \left\{\frac{1}{8} (1-3 v) v,\frac{1}{8} v (v+1),\frac{1}{8} v (v+1),\frac{1}{8} v(v+1)\right\}$, which means that the map \eqref{WernerMap} detects entanglement of a Werner state for $v>\tfrac{1}{3}$. On the other hand we may apply the condition \eqref{MainCondMap} with the map \eqref{WernerMap} to a singlet state, which gives an operator with
eigenvalues $\left\{-\frac{v}{4},\frac{v}{4},\frac{v}{4},\frac{v}{4}\right\}$. We see that the map generated by a Werner state \eqref{WernerMap} can detect entanglement of a singlet state for all values of $v$.

Although the above examples of positive maps are very simple, we see that a different choice of an initial state generates different map for the same metric. Can we choose the metric tensor in a way which would lead to the same map for arbitrary state? We will show that such choice is possible only for a nonlinear and non-positive map $\G$.

\subsection{New nonlinear geometric entanglement criterion}

In the previous section we discussed several examples of entanglement indicators \eqref{MainCond}
generated by some linear maps $\G$. They led us to experimentally friendly entanglement criteria, which however were not optimal for discussed classes of states. In this section we introduce a new entanglement indicator, based on the following explicitly \emph{nonlinear} map $\G$:
\be
\label{GSign}
\G[\rho]=\frac{1}{4}\sum_{kl} \left(\sum_{ij}G^{ij}_{kl}\textrm{sgn}(T_{ij})\right)\sigma^A_k\ten\sigma^B_l,
\ee
which depends only on the sign of the elements of the correlation tensor, but not on its values. The metric tensor in the above formula is assumed to be the standard one \eqref{StdMet}. The entanglement identifier \eqref{MainCond} generated by the map \eqref{GSign} has the following direct form:
\be
\frac{\sqrt{d_Ad_B}}{2}\max_{X^A,Y^B}\sum_{i,j=1}^{d_Ad_B-1}X^A_iY^B_j\textrm{sgn}(T_{ij})<\sum_{i,j=1}^{d_Ad_B-1}|T_{ij}|,\non
\label{SignCond}
\ee
in which $X^A$ and $Y^B$ are local Bloch vectors. The above condition can be treated as a geometric entanglement identifier with a specific choice of the metric, namely the RHS of \eqref{SignCond} defines a norm called the Manhattan norm, which gives rise to a distance measure on the space of correlation tensors called the Manhattan (or taxicab) distance.

The entanglement identifier \eqref{SignCond} turns out to be extremely useful. In the case of two discussed families of two-qubit states (see Fig. \ref{BDFig} and \ref{WSFig}) it perfectly detects entanglement of all the states from both the families. This is in sharp contrast with nonlinear entanglement indicators proposed in \cite{Moroder08} which achieve the same aim in the limit of infinitely many improvements to the entanglement witness. The discussed identifier \eqref{SignCond} is also useful for higher dimensional systems. Let us discuss the following family of two two-qutrit Werner-type states:
\be
\rho_{QtW}=v\proj{\psi(\alpha,\beta)}-\frac{1-v}{9}\matone,
\label{QtritWerStates}
\ee
in which the pure two-qutrit state $\psi(\alpha,\beta)$ is given by:
\be
\ket{\psi(\alpha,\beta)}&=&\sin(\alpha)\cos(\beta)\ket{00}+\sin(\alpha)\sin(\beta)\ket{11}\non
&&+\cos(\alpha)\ket{22}.
\ee
If we fix $\beta=\pi/4$ we get a two-parameter family of states $\rho_{QtW}(v,\alpha)$, which can be represented on a plane with two orthogonal axes corresponding to values of $v$ and $\alpha$. The entanglement of this family of states can be uniquely detected via PPT condition \cite{Pittenger00}. It turns out, that experimentally friendly entanglement identifier \eqref{SignCond} is not optimal, however it is much better than the one based on linear map $\G$ \eqref{StdMet} (see Fig. \ref{QTWerFig} for comparison).

\begin{figure}
\includegraphics[width=\columnwidth]{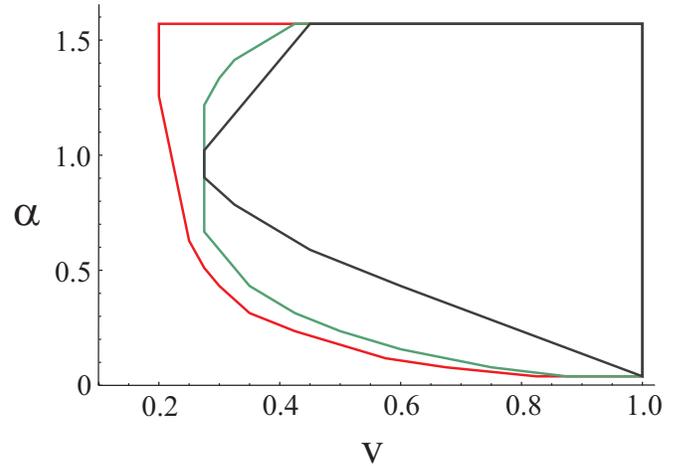}
\caption{Entanglement detection of two-qutrit Werner states \eqref{QtritWerStates}. The set of all entangled states is depicted by the outermost region (red line). The set of entangled states detected by the condition \eqref{SignCond} is depicted by the middle region (green line), whereas the set of entangled states detected by the condition \eqref{MainCond} with a standard metric \eqref{StdMet} is depicted by the inner region (black line).}
\label{QTWerFig}
\end{figure}

\subsection{Nonlinear entanglement indentifiers and the PPT criterion}

In the previous sections we discussed how to transform the functional entanglement identifier \eqref{MainCond} into a positive-map-based condition \eqref{MainCondMap}. The positive map \eqref{MainMap}  corresponding to a given functional identifier has a very peculiar form, namely it depends both on the map $\G$ generating the identifier, and on the state on which the identifier acts. The question arises if we can to some extent reverse the construction, namely fix the form of a positive map, and search for a definite form of $\G$ which would generate this map for every state $\rho$. The most interesting case is the PPT map $\openone\ten\mathcal T$, which for two qubits and qubit-qutrit systems is a universal entanglement identifier. Note that our map-based condition \eqref{MainCondMap} is equivalent to the PPT condition, if and only if the map $\Lambda^{\textrm d}_{\G,\rho}$ \eqref{MainMap} is an identity map. Can we find a form of $\G$ which would generate such a map for every state? In general solving such a problem would require solving a system of highly nonlinear equations, however in the case of two-qubit systems the solution can be quite easily guessed:
\begin{propo}
\label{PropPPT}
Let us assume that $d_A=d_B=2$, and we choose the same operator basis $\{\sigma_i\}$ for both subsystems.
Then the choice of a nonlinear map $\G$ \eqref{GNonLin} with tensor $G^{ij}_{kl}=-\delta_{ij}\delta_{jl}\delta_{ik}[\{i,j,k,l\}\neq 0]$, and a constant function  $f(x)=\textrm{const}(x)=1$:
\be
\label{GPPT}
\G_{\textrm{PPT}}[\rho]&=&-\frac{1}{4}\sum_{kl=1}^{3} \left(\sum_{ij=1}^{3}\delta_{ij}\delta_{jl}\delta_{ik}\textrm{const}(T_{ij})\right)\sigma^k\ten\sigma^l\non
&=&-\frac{1}{4}\sum_{i=1}^{3}\sigma_i\ten\sigma_i,
\ee
generates an identity map \eqref{MainMap}: $$\Lambda^{\textrm d}_{\G_{\textrm{PPT}},\rho}[\lambda]=\lambda$$ for any $\rho$ and $\lambda$.
\end{propo}
\begin{proof}
See Appendix \ref{ProofPPT}.
\end{proof}
Note that the map $\G_{\textrm{PPT}}$ is neither linear nor generates it a metric \eqref{GPos}. It gives rise to the following functional entanglement indicator \eqref{MainCond}:
\be
\sum_{i=1}^3 T_{ii}<-1.
\ee
The above indicator detects entanglement of a singlet state (for which $T_{ii}=-1$ and therefore the LHS is $-3$), however it fails to detect entanglement of other Bell states. On the contrary the PPT map detects entanglement of all two-qubit states.

\section{The case of multipartite entanglement}

Insofar we discussed the bipartite entanglement. However, as discussed in \cite{LMPZ11}, the metric-based entanglement criteria \eqref{tensormulti} work for detection of arbitrary level of partial separability in the multipartite case. One of the most important conclusions from \cite{LMPZ11} is that in order to prove a genuine entanglement one needs to reject biseparability with respect to all possible bipartitions using the metric-based conditions \eqref{MainCondGLin}. Is it possible to use all the introduced framework of non-linear functional entanglement identifiers to the multipartite case? The answer is positive, however in such a generalization we loose one important property. Namely the connection with positive-map based condition fails in this case. Let us discuss it in more details in the case of checking genuine entanglement by rejecting biseparability. One can still use the relation \eqref{MapGLin} to construct a positive-map-based criterion for any bipartition. The problem is that the functional identifier \eqref{MainCond} applied to all bipartitions detects only genuine entanglement, whereas the positive map \eqref{MapGLin} corresponding to it can detect also a bipartite entanglement, which is not interesting in this case.  Let us show it directly on the example of a three-qubit W state $| W \rangle = \frac{1}{\sqrt{3}} (|100 \rangle + |010 \rangle + |001 \rangle)$ mixed with a white noise:
\be
\rho_W(v)=v\proj{W}+\frac{1-v}{8}\matone.
\label{Wnoise}
\ee
Since the above state is symmetric with respect to all subsystems, the values of the entanglement identifiers must be the same for all bipartitions. Let us then fix the bipartition $\mathcal A\mathcal B|\mathcal C$, where $\mathcal A,\mathcal B,\mathcal C$ represent state spaces of the consecutive qubits. According to \cite{LMPZ11} the entanglement identifier \eqref{MainCondGLin} with a standard metric tensor $G^{i_1,i_2,i_3}_{j_1,j_2,j_3}=\delta_{i_1,j_1}\delta_{i_2,j_2}\delta_{i_3,j_3}$ rejects biproduct character of the state within a bipartition $\mathcal A\mathcal B|\mathcal C$ for $v>0.636$. Since rejecting a biproduct character of a state with respect to all bipartitions using the metric-based identifiers is equivalent to rejecting a biseparability of a state, the metric-based condition detects genuine 3-partite entanglement of a state \eqref{Wnoise} for $v>0.636$. However, this is not the optimal value, which is given by $v>0.521$ \cite{Moroder11}. Now let us consider the map-based version of this metric identifier \eqref{MapGLin}, with subsystem $A=\mathcal A\mathcal B$ and $B=\mathcal C$. As a basis for subsystem $A$ we choose the product basis $\sigma_{i(i_1,i_2)}^A=\sigma_{i_1}^{\mathcal A}\ten \sigma_{i_2}^{\mathcal B}$. Then the map \eqref{MapGLin} takes the following form:
\be
\Lambda^{\textrm d}_{\hat G,\rho_W(v)}[\lambda]=\frac{1}{4}\tilde T_{00}\matone^{\mathcal A\mathcal B}-\sum_{i_1,i_2,j=1}^3 T_{i_1i_2j}\Tr(\sigma_j^{\mathcal C}\lambda)\sigma_{i_1}^{\mathcal A}\ten \sigma_{i_2}^{\mathcal B}.\non
\label{MapGLinW}
\ee
The map-based condition \eqref{MainCondMap} using the above map \eqref{MapGLinW} to the state \eqref{Wnoise} detects entanglement for $v>0.456$. Since this value is below the genuine entanglement value $v=0.521$, it is clear that the map \eqref{MapGLinW} detects a bipartite entanglement present in the state \eqref{Wnoise}. The reason for this is that the map \eqref{MapGLinW} detects bipartite entanglement for two pairs of subsystems $\{\mathcal A,\mathcal C\}$ and $\{\mathcal B,\mathcal C\}$. This entanglement is not detected by the metric-based condition, since the entire state $\rho_{W}(v<0.521)$ --- although containing bipartite entanglement --- is biseparable, and the metric-based condition \eqref{MainCond} for detecting genuine entanglement is by definition positive on biseparable states. Therefore detecting genuine entanglement with map-based conditions demands a redefinition of a condition \eqref{MainCondMap} in order to make the LHS of this condition positive for biseparable states \cite{Huber14, Clivas16}.

\section{Discussion}
In this note we presented a generalized purely algebraic approach to nonlinear entanglement identifiers, which in many cases turns out to be experimentally friendly. We showed that such identifiers can be always transformed to a positive map-based versions. The map-based versions are similar in form to the former known map criteria due to Horodecki \cite{HHH96, HHH01}, however the maps are explicitly adjusted to the investigated states. This property implies, that the final criteria involve nonlinear functions of investigated states. The provided relations allow for explicit constructions of positive maps which can detect entanglement of a broad class of states. It would be interesting to find conditions for the functional identifiers under which the corresponding positive maps have important specific properties, like decomposability or extremality.

\section{Acknowledgements}
We thank Marek Ku\'s and Karol \.Zyczkowski for inspiring discussions. MM is supported by the  
National Science Centre, Poland, grant number 2015/16/S/ST2/00447 within the 
project FUGA 4 for postdoctoral training. AR, AK and WL are supported by the  National Science Centre, Poland, Grant No. 2014/14/M/ST2/00818.
ZP is supported by the  National Science Centre, Poland, grant number 
2014/15/B/ST6/05204.

\appendix
\section{Choi-Jamiolkowski isomorphism for arbitrary vector spaces}
\label{ApChoi}
We introduce a version of the Choi-Jamiolkowski isomorphism \cite{Jamiolkowski72, Choi75, Choi82}, which works for any vector spaces of finite dimension. Let us assume $\Ham_A$ and $\Ham_B$ are two finite dimensional vector spaces enowed with a scalar product:
\be
|v_1)\cdot |v_2)\equiv (v_1|v_2).
\ee
We will use a bra-ket type notation, although we need not specify to which physical objects the vectors $|v_1)$ and $|v_2)$ correspond. There exists a unique algebraic tensor product $\Ham_A \ten \Ham_B$, which is the freest vector space composed of formal objects $|v_1)\ten|v_2)$ that is bilinear (linear in both factors). We also assume that for every vector $|v)$ there exits its dual $(v|\in \Ham^*$, which plays the role of a linear functional on the space $\Ham$:
\be
(v|(x)=(v|x).
\ee
\begin{propo}
There exists a natural isomorphism between the vector  spaces $\Ham_A\ten\Ham_B$, $\Ham_B\ten\Ham_A^*$ and $\Ll(\Ham_A\mapsto\Ham_B)$,  represented by the objects:
\be
\sum_{i,j}v_{ij}|i)\ten|j),\,\sum_{i,j}v_{ij}|j)\ten(i|,\,\sum_{i,j}v_{ij}(i|\cdot)|j),
\ee
which is specified by the bijections:
\be
\sum_{i,j}v_{ij}|i)\ten|j)&=&\left(\openone\ten\sum_{i,j}v_{ij}|j)\ten(i|\right)\sum_{m}|m)\ten|m),\non
&=&\left(\openone\ten\sum_{i,j}v_{ij}(i|\cdot)|j)\right)\sum_{m}|m)\ten|m),\non
\label{iso1}
\ee 
in which the object $\sum_{m}|m)\ten|m)$ is  the \emph{maximally entangled} element of $\Ham_A\ten\Ham_A$.
\label{propoiso1}
\end{propo}
\begin{proof}
All the vector spaces $\Ham_A\ten\Ham_B$, $\Ham_A\ten\Ham_B^*$ and $\Ll(\Ham_B\mapsto\Ham_A)$ have dimension $\operatorname{dim}(\Ham_A)\operatorname{dim}(\Ham_B)$, therefore are isomorphic. The special forms of bijection between them provided by \eqref{iso1} hold trivially:
\be
&&\left(\openone\ten\sum_{i,j}v_{ij}|j)\ten(i|\right)\sum_{m}|m)\ten|m),\non
&&=\left(\openone\ten\sum_{i,j}v_{ij}(i|\cdot)|j)\right)\sum_{m}|m)\ten|m)\non
&&=\sum_m |m)\ten\sum_{i,j}v_{ij}(i|m)|j)\non
&&=\sum_{m,i,j}v_{ij}(i|m) |m)\ten |j)=\sum_{m,i,j}v_{ij}\delta_{im} |m)\ten |j)\non
&&=\sum_{i,j}v_{ij}|i)\ten |j).
\label{iso2}
\ee 
In the first step we use the natural identification of $\Ham_B\ten\mathbb C$ with $\Ham_B$.
\end{proof}

\section{Proof of Proposition \ref{propLdual}}
\label{proofLdual}
\begin{proof}
\be
\Tr\left(\alpha_1\Lambda_{\G,\rho}[\alpha_2]\right)&=&\Tr\left(\alpha_1\sum_{ij} w_{ij}^{\G,\rho}\,\Tr\left(\sigma^A_i\alpha_2\right)\sigma^B_j\right)\non
&=&\sum_{ij} w_{ij}^{\G,\rho}\,\Tr\left(\alpha_1\sigma^B_j\right)\Tr\left(\sigma^A_i\alpha_2\right).
\ee
On the other hand:
\be
\Tr\left(\Lambda_{\G,\rho}^d[\alpha_1]\alpha_2\right)&=&\Tr\left(\sum_{ij} w_{ij}^{\G,\rho}\,\Tr\left(\sigma^B_j\alpha_1\right)\sigma^{A}_i\alpha_2\right)\non
&=&\sum_{ij} w_{ij}^{\G,\rho}\,\Tr\left(\alpha_1\sigma^B_j\right)\Tr\left(\sigma^A_i\alpha_2\right).
\ee
\end{proof}


\section{Proof of Proposition \ref{propFiPlus}}
\label{proofFiPlus}
\begin{proof}
First let us note, that  due to the fact $\ts^2=\openone$, the formula \eqref{fiplusone} is equivalent to:
\be
\sum_{m}\sigma_m\ten \sigma_m^T=\rho_{\Phi^+}.
\label{fiplusone1}
\ee
Since $\{ \sigma_i \}_{i=1}^{d^2}$ forms a Hermitian orthonormal basis of $d 
\times d $ matrices, then vectors $\{ \vket{\sigma_i} \}_{i=1}^{d^2}$ form an 
orthonormal basis of $\mathbb{C}^{d^2} = \mathbb{C}^{d}\otimes \mathbb{C}^{d}$,
where the vector $\vket{\rho}$ denotes vectorized (reshaped in lexicographical 
order) version of a matrix $\rho$.
The orthonormality and completeness ensures us that  
\begin{equation}
\sum_m \vket{\sigma_m} \vbra{\sigma_m} = \matone_{d^2},
\end{equation} 
which in terms of matrix elements in computational basis reads
\begin{equation}
\begin{split}
&\bra{ij}\sum_m \vket{\sigma_m} \vbra{\sigma_m} {kl}\rangle =
\sum_m \langle{ij} \vket{\sigma_m} \vbra{\sigma_m} {kl}\rangle \\
&= \sum_m \bra{i} \sigma_m \ket{j} \bra{k} \overline{\sigma}_m \ket{l} 
= \sum_m \bra{i} \sigma_m \ket{j} \bra{k} \sigma_m^T \ket{l} 
\\
&= \delta_{ik} \delta_{jl}.
\end{split}
\end{equation} 
Now, using the above we write elements of matrix $\sum_{m} \sigma_m \ten 
\sigma_m^T$ and show that they are equal to elements of matrix $\rho_{\Phi^+}$,
\begin{eqnarray}
&& \bra{ij}\sum_{m} \sigma_m \ten \sigma_m^T \ket{kl} = 
\sum_m \bra{i} \sigma_m \ket{k} \bra{j} \sigma_m^T \ket{l} = 
\delta_{ij}\delta_{kl}  \nonumber \\ 
&& = \bra{ij} \left(\sum_{m n} \ket{mm}\bra{nn}\right)  \ket{kl} 
= \bra{ij}  \rho_{\Phi^+} \ket{kl}.
\end{eqnarray}
Which gives us the result.
\end{proof}

\section{Proof of Proposition \ref{PropGTMax}}
\label{ProofGTMax}
\begin{proof}
It suffices to show that:
\be
w_{ij}^{\G,\rho}=\tilde T_{ij}-\sum_{kl}G^{kl}_{ij}T_{kl},
\ee
in which the only nonzero element of the tensor $\tilde T$ is $\tilde T_{00}=\frac{\sqrt{d_Ad_B}}{2}\max_{\hat X,\hat Y}\sum_{klmn}X_kY_lG^{kl}_{mn}T_{mn}$.
Taking the definition of $w_{ij}^{\G,\rho}$ we have:
\be
w_{ij}^{\G,\rho}&=&\Tr(W_{\G}[\rho]\sigma^A_i\ten\sigma^B_j)=\Tr((\tilde\omega_0\matone-\G[\rho])\sigma^A_i\ten\sigma^B_j)\non
&=&\tilde\omega_0\Tr(\sigma^A_0)\Tr(\sigma^B_0)-\Tr(\G[\rho]\sigma^A_i\ten\sigma^B_j)\non
&=&2\sqrt{d_Ad_B}\tilde\omega_0-\sum_{kl}G^{kl}_{ij}T_{kl}\non
&=&\frac{\sqrt{d_Ad_b}}{2}\max_{\hat X,\hat Y}\sum_{klmn}X_kY_lG^{kl}_{mn}T_{mn}-\sum_{kl}G^{kl}_{ij}T_{kl}.\non
\ee

\end{proof}

\section{Proof of Proposition \ref{PropPPT}}
\label{ProofPPT}
\begin{proof}
For $\Ham_A=\Ham_B$ and a choice of local operator basis in the form $\{\sigma_i\}$ the map \eqref{MainMap} is an identity map if and only if $w_{ij}^{\G,\rho}=\delta_{ij}$. Indeed, for any $\lambda$ the action of this map reads:
\be
\Lambda_{\G,\rho}^d[\lambda]&=&\sum_{ij} \frac{1}{4}w_{ij}^{\G,\rho}\,\Tr\left(\sigma_j\lambda\right)\sigma_i=\sum_{ij} \frac{1}{4}\delta_{ij}\,\Tr\left(\sigma_j\lambda\right)\sigma_i\non
&=&\sum_{i}\Tr\left(\frac{1}{4}\sigma_i\lambda\right)\sigma_i=\lambda
\ee
and represents a decomposition of $\lambda$ in the operator basis $\{\sigma_i\}$. Therefore it suffices to show that a choice of the map $\G_{\textrm{PPT}}$ in the form \eqref{GPPT} guarantees that $w_{ij}^{\G,\rho}=\delta_{ij}$. Let us check it directly:
\be
w_{ij}^{\G,\rho}&=&\Tr(W_{\G}[\rho]\sigma_i\ten\sigma_j)=\Tr((\tilde\omega_0(\rho)\openone-\G[\rho])\sigma_i\ten\sigma_j)\non
&=&\max_{\hat X,\hat Y}\left(-\sum_{klmn}X_kY_l\right)[\{i,j\}=0]+\sum_{kl}\delta_{ij}\delta_{jl}\delta_{ik}\non
&=&[\{i,j\}=0]+\delta_{ij}[\{i,j\}\neq 0]=\delta_{ij}.
\ee
\end{proof}


\begin{thebibliography}{26}%
	\makeatletter
	\providecommand \@ifxundefined [1]{%
		\@ifx{#1\undefined}
	}%
	\providecommand \@ifnum [1]{%
		\ifnum #1\expandafter \@firstoftwo
		\else \expandafter \@secondoftwo
		\fi
	}%
	\providecommand \@ifx [1]{%
		\ifx #1\expandafter \@firstoftwo
		\else \expandafter \@secondoftwo
		\fi
	}%
	\providecommand \natexlab [1]{#1}%
	\providecommand \enquote  [1]{``#1''}%
	\providecommand \bibnamefont  [1]{#1}%
	\providecommand \bibfnamefont [1]{#1}%
	\providecommand \citenamefont [1]{#1}%
	\providecommand \href@noop [0]{\@secondoftwo}%
	\providecommand \href [0]{\begingroup \@sanitize@url \@href}%
	\providecommand \@href[1]{\@@startlink{#1}\@@href}%
	\providecommand \@@href[1]{\endgroup#1\@@endlink}%
	\providecommand \@sanitize@url [0]{\catcode `\\12\catcode `\$12\catcode
		`\&12\catcode `\#12\catcode `\^12\catcode `\_12\catcode `\%12\relax}%
	\providecommand \@@startlink[1]{}%
	\providecommand \@@endlink[0]{}%
	\providecommand \url  [0]{\begingroup\@sanitize@url \@url }%
	\providecommand \@url [1]{\endgroup\@href {#1}{\urlprefix }}%
	\providecommand \urlprefix  [0]{URL }%
	\providecommand \Eprint [0]{\href }%
	\providecommand \doibase [0]{http://dx.doi.org/}%
	\providecommand \selectlanguage [0]{\@gobble}%
	\providecommand \bibinfo  [0]{\@secondoftwo}%
	\providecommand \bibfield  [0]{\@secondoftwo}%
	\providecommand \translation [1]{[#1]}%
	\providecommand \BibitemOpen [0]{}%
	\providecommand \bibitemStop [0]{}%
	\providecommand \bibitemNoStop [0]{.\EOS\space}%
	\providecommand \EOS [0]{\spacefactor3000\relax}%
	\providecommand \BibitemShut  [1]{\csname bibitem#1\endcsname}%
	\let\auto@bib@innerbib\@empty
	\bibitem [{\citenamefont {G\"uhne}\ and\ \citenamefont
		{T\'oth}(2009)}]{Guhne09}%
	\BibitemOpen
	\bibfield  {author} {\bibinfo {author} {\bibfnamefont {O.}~\bibnamefont
			{G\"uhne}}\ and\ \bibinfo {author} {\bibfnamefont {G.}~\bibnamefont
			{T\'oth}},\ }\href {\doibase https://doi.org/10.1016/j.physrep.2009.02.004}
	{\bibfield  {journal} {\bibinfo  {journal} {Physics Reports}\ }\textbf
		{\bibinfo {volume} {474}},\ \bibinfo {pages} {1 } (\bibinfo {year}
		{2009})}\BibitemShut {NoStop}%
	\bibitem [{\citenamefont {Horodecki}\ \emph {et~al.}(2009)\citenamefont
		{Horodecki}, \citenamefont {Horodecki}, \citenamefont {Horodecki},\ and\
		\citenamefont {Horodecki}}]{HHH09}%
	\BibitemOpen
	\bibfield  {author} {\bibinfo {author} {\bibfnamefont {R.}~\bibnamefont
			{Horodecki}}, \bibinfo {author} {\bibfnamefont {P.}~\bibnamefont
			{Horodecki}}, \bibinfo {author} {\bibfnamefont {M.}~\bibnamefont
			{Horodecki}}, \ and\ \bibinfo {author} {\bibfnamefont {K.}~\bibnamefont
			{Horodecki}},\ }\href {\doibase 10.1103/RevModPhys.81.865} {\bibfield
		{journal} {\bibinfo  {journal} {Rev. Mod. Phys.}\ }\textbf {\bibinfo {volume}
			{81}},\ \bibinfo {pages} {865} (\bibinfo {year} {2009})}\BibitemShut
	{NoStop}%
	\bibitem [{\citenamefont {Horodecki}\ \emph {et~al.}(1996)\citenamefont
		{Horodecki}, \citenamefont {Horodecki},\ and\ \citenamefont
		{Horodecki}}]{HHH96}%
	\BibitemOpen
	\bibfield  {author} {\bibinfo {author} {\bibfnamefont {M.}~\bibnamefont
			{Horodecki}}, \bibinfo {author} {\bibfnamefont {P.}~\bibnamefont
			{Horodecki}}, \ and\ \bibinfo {author} {\bibfnamefont {R.}~\bibnamefont
			{Horodecki}},\ }\href@noop {} {\bibfield  {journal} {\bibinfo  {journal}
			{Phys. Lett. A}\ }\textbf {\bibinfo {volume} {223}},\ \bibinfo {pages} {1}
		(\bibinfo {year} {1996})}\BibitemShut {NoStop}%
	\bibitem [{\citenamefont {Horodecki}\ \emph {et~al.}(2001)\citenamefont
		{Horodecki}, \citenamefont {Horodecki},\ and\ \citenamefont
		{Horodecki}}]{HHH01}%
	\BibitemOpen
	\bibfield  {author} {\bibinfo {author} {\bibfnamefont {M.}~\bibnamefont
			{Horodecki}}, \bibinfo {author} {\bibfnamefont {P.}~\bibnamefont
			{Horodecki}}, \ and\ \bibinfo {author} {\bibfnamefont {R.}~\bibnamefont
			{Horodecki}},\ }\href@noop {} {\bibfield  {journal} {\bibinfo  {journal}
			{Phys. Lett. A}\ }\textbf {\bibinfo {volume} {283}},\ \bibinfo {pages} {1}
		(\bibinfo {year} {2001})}\BibitemShut {NoStop}%
	\bibitem [{\citenamefont {Huber}\ and\ \citenamefont
		{Sengupta}(2014)}]{Huber14}%
	\BibitemOpen
	\bibfield  {author} {\bibinfo {author} {\bibfnamefont {M.}~\bibnamefont
			{Huber}}\ and\ \bibinfo {author} {\bibfnamefont {R.}~\bibnamefont
			{Sengupta}},\ }\href {\doibase 10.1103/PhysRevLett.113.100501} {\bibfield
		{journal} {\bibinfo  {journal} {Phys. Rev. Lett.}\ }\textbf {\bibinfo
			{volume} {113}},\ \bibinfo {pages} {100501} (\bibinfo {year}
		{2014})}\BibitemShut {NoStop}%
	\bibitem [{\citenamefont {W\"olk}\ \emph {et~al.}(2014)\citenamefont {W\"olk},
		\citenamefont {Huber},\ and\ \citenamefont {G\"uhne}}]{Wolk14}%
	\BibitemOpen
	\bibfield  {author} {\bibinfo {author} {\bibfnamefont {S.}~\bibnamefont
			{W\"olk}}, \bibinfo {author} {\bibfnamefont {M.}~\bibnamefont {Huber}}, \
		and\ \bibinfo {author} {\bibfnamefont {O.}~\bibnamefont {G\"uhne}},\ }\href
	{\doibase 10.1103/PhysRevA.90.022315} {\bibfield  {journal} {\bibinfo
			{journal} {Phys. Rev. A}\ }\textbf {\bibinfo {volume} {90}},\ \bibinfo
		{pages} {022315} (\bibinfo {year} {2014})}\BibitemShut {NoStop}%
	\bibitem [{\citenamefont {Clivaz}\ \emph {et~al.}(2017)\citenamefont {Clivaz},
		\citenamefont {Huber}, \citenamefont {Lami},\ and\ \citenamefont
		{Murta}}]{Clivas16}%
	\BibitemOpen
	\bibfield  {author} {\bibinfo {author} {\bibfnamefont {F.}~\bibnamefont
			{Clivaz}}, \bibinfo {author} {\bibfnamefont {M.}~\bibnamefont {Huber}},
		\bibinfo {author} {\bibfnamefont {L.}~\bibnamefont {Lami}}, \ and\ \bibinfo
		{author} {\bibfnamefont {G.}~\bibnamefont {Murta}},\ }\href@noop {}
	{\bibfield  {journal} {\bibinfo  {journal} {J. Math. Phys.}\ }\textbf
		{\bibinfo {volume} {58}},\ \bibinfo {pages} {082201} (\bibinfo {year}
		{2017})}\BibitemShut {NoStop}%
	\bibitem [{\citenamefont {Laskowski}\ \emph {et~al.}(2011)\citenamefont
		{Laskowski}, \citenamefont {Markiewicz}, \citenamefont {Paterek},\ and\
		\citenamefont {\ifmmode~\dot{Z}\else \.{Z}\fi{}ukowski}}]{LMPZ11}%
	\BibitemOpen
	\bibfield  {author} {\bibinfo {author} {\bibfnamefont {W.}~\bibnamefont
			{Laskowski}}, \bibinfo {author} {\bibfnamefont {M.}~\bibnamefont
			{Markiewicz}}, \bibinfo {author} {\bibfnamefont {T.}~\bibnamefont {Paterek}},
		\ and\ \bibinfo {author} {\bibfnamefont {M.}~\bibnamefont
			{\ifmmode~\else \.{Z}\fi{}ukowski}},\ }\href {\doibase
		10.1103/PhysRevA.84.062305} {\bibfield  {journal} {\bibinfo  {journal} {Phys.
				Rev. A}\ }\textbf {\bibinfo {volume} {84}},\ \bibinfo {pages} {062305}
		(\bibinfo {year} {2011})}\BibitemShut {NoStop}%
	\bibitem [{\citenamefont {Laskowski}\ \emph {et~al.}(2012)\citenamefont
		{Laskowski}, \citenamefont {Richart}, \citenamefont {Schwemmer},
		\citenamefont {Paterek},\ and\ \citenamefont {Weinfurter}}]{Laskowski12}%
	\BibitemOpen
	\bibfield  {author} {\bibinfo {author} {\bibfnamefont {W.}~\bibnamefont
			{Laskowski}}, \bibinfo {author} {\bibfnamefont {D.}~\bibnamefont {Richart}},
		\bibinfo {author} {\bibfnamefont {C.}~\bibnamefont {Schwemmer}}, \bibinfo
		{author} {\bibfnamefont {T.}~\bibnamefont {Paterek}}, \ and\ \bibinfo
		{author} {\bibfnamefont {H.}~\bibnamefont {Weinfurter}},\ }\href {\doibase
		10.1103/PhysRevLett.108.240501} {\bibfield  {journal} {\bibinfo  {journal}
			{Phys. Rev. Lett.}\ }\textbf {\bibinfo {volume} {108}},\ \bibinfo {pages}
		{240501} (\bibinfo {year} {2012})}\BibitemShut {NoStop}%
	\bibitem [{\citenamefont {Markiewicz}\ \emph {et~al.}(2013)\citenamefont
		{Markiewicz}, \citenamefont {Laskowski}, \citenamefont {Paterek},\ and\
		\citenamefont {\ifmmode~\else \.{Z}\fi{}ukowski}}]{Markiewicz13}%
	\BibitemOpen
	\bibfield  {author} {\bibinfo {author} {\bibfnamefont {M.}~\bibnamefont
			{Markiewicz}}, \bibinfo {author} {\bibfnamefont {W.}~\bibnamefont
			{Laskowski}}, \bibinfo {author} {\bibfnamefont {T.}~\bibnamefont {Paterek}},
		\ and\ \bibinfo {author} {\bibfnamefont {M.}~\bibnamefont
			{\ifmmode~\else \.{Z}\fi{}ukowski}},\ }\href {\doibase
		10.1103/PhysRevA.87.034301} {\bibfield  {journal} {\bibinfo  {journal} {Phys.
				Rev. A}\ }\textbf {\bibinfo {volume} {87}},\ \bibinfo {pages} {034301}
		(\bibinfo {year} {2013})}\BibitemShut {NoStop}%
	\bibitem [{\citenamefont {Laskowski}\ \emph {et~al.}(2013)\citenamefont
		{Laskowski}, \citenamefont {Markiewicz}, \citenamefont {Paterek},\ and\
		\citenamefont {Weinar}}]{Laskowski13}%
	\BibitemOpen
	\bibfield  {author} {\bibinfo {author} {\bibfnamefont {W.}~\bibnamefont
			{Laskowski}}, \bibinfo {author} {\bibfnamefont {M.}~\bibnamefont
			{Markiewicz}}, \bibinfo {author} {\bibfnamefont {T.}~\bibnamefont {Paterek}},
		\ and\ \bibinfo {author} {\bibfnamefont {R.}~\bibnamefont {Weinar}},\ }\href
	{\doibase 10.1103/PhysRevA.88.022304} {\bibfield  {journal} {\bibinfo
			{journal} {Phys. Rev. A}\ }\textbf {\bibinfo {volume} {88}},\ \bibinfo
		{pages} {022304} (\bibinfo {year} {2013})}\BibitemShut {NoStop}%
	\bibitem [{\citenamefont {Laskowski}\ \emph {et~al.}(2015)\citenamefont
		{Laskowski}, \citenamefont {Markiewicz}, \citenamefont {Rosseau},
		\citenamefont {Byrnes}, \citenamefont {Kostrzewa},\ and\ \citenamefont
		{Ko\l{}odziejski}}]{Laskowski15}%
	\BibitemOpen
	\bibfield  {author} {\bibinfo {author} {\bibfnamefont {W.}~\bibnamefont
			{Laskowski}}, \bibinfo {author} {\bibfnamefont {M.}~\bibnamefont
			{Markiewicz}}, \bibinfo {author} {\bibfnamefont {D.}~\bibnamefont {Rosseau}},
		\bibinfo {author} {\bibfnamefont {T.}~\bibnamefont {Byrnes}}, \bibinfo
		{author} {\bibfnamefont {K.}~\bibnamefont {Kostrzewa}}, \ and\ \bibinfo
		{author} {\bibfnamefont {A.}~\bibnamefont {Ko\l{}odziejski}},\ }\href
	{\doibase 10.1103/PhysRevA.92.022339} {\bibfield  {journal} {\bibinfo
			{journal} {Phys. Rev. A}\ }\textbf {\bibinfo {volume} {92}},\ \bibinfo
		{pages} {022339} (\bibinfo {year} {2015})}\BibitemShut {NoStop}%
	\bibitem [{\citenamefont {Badzia\ifmmode~\mbox{\c{}}\else \c{}\fi{}g}\ \emph
		{et~al.}(2008)\citenamefont {Badzia\ifmmode~\mbox{\c{}}\else \c{}\fi{}g},
		\citenamefont {Brukner}, \citenamefont {Laskowski}, \citenamefont {Paterek},\
		and\ \citenamefont {\ifmmode~\else \.{Z}\fi{}ukowski}}]{BBLPZ08}%
	\BibitemOpen
	\bibfield  {author} {\bibinfo {author} {\bibfnamefont {P.}~\bibnamefont
			{Badzia\ifmmode~\mbox{\c{}}\else \c{}\fi{}g}}, \bibinfo {author}
		{\bibfnamefont {C.}~\bibnamefont {Brukner}}, \bibinfo {author} {\bibfnamefont
			{W.}~\bibnamefont {Laskowski}}, \bibinfo {author} {\bibfnamefont
			{T.}~\bibnamefont {Paterek}}, \ and\ \bibinfo {author} {\bibfnamefont
			{M.}~\bibnamefont {\ifmmode~\else \.{Z}\fi{}ukowski}},\ }\href
	{\doibase 10.1103/PhysRevLett.100.140403} {\bibfield  {journal} {\bibinfo
			{journal} {Phys. Rev. Lett.}\ }\textbf {\bibinfo {volume} {100}},\ \bibinfo
		{pages} {140403} (\bibinfo {year} {2008})}\BibitemShut {NoStop}%
	\bibitem [{\citenamefont {G\"uhne}\ and\ \citenamefont
		{L\"utkenhaus}(2007)}]{Guhne07}%
	\BibitemOpen
	\bibfield  {author} {\bibinfo {author} {\bibfnamefont {O.}~\bibnamefont
			{G\"uhne}}\ and\ \bibinfo {author} {\bibfnamefont {N.}~\bibnamefont
			{L\"utkenhaus}},\ }\href {http://stacks.iop.org/1742-6596/67/i=1/a=012004}
	{\bibfield  {journal} {\bibinfo  {journal} {Journal of Physics: Conference
				Series}\ }\textbf {\bibinfo {volume} {67}},\ \bibinfo {pages} {012004}
		(\bibinfo {year} {2007})}\BibitemShut {NoStop}%
	\bibitem [{\citenamefont {Moroder}\ \emph {et~al.}(2008)\citenamefont
		{Moroder}, \citenamefont {G\"uhne},\ and\ \citenamefont
		{L\"utkenhaus}}]{Moroder08}%
	\BibitemOpen
	\bibfield  {author} {\bibinfo {author} {\bibfnamefont {T.}~\bibnamefont
			{Moroder}}, \bibinfo {author} {\bibfnamefont {O.}~\bibnamefont {G\"uhne}}, \
		and\ \bibinfo {author} {\bibfnamefont {N.}~\bibnamefont {L\"utkenhaus}},\
	}\href {\doibase 10.1103/PhysRevA.78.032326} {\bibfield  {journal} {\bibinfo
			{journal} {Phys. Rev. A}\ }\textbf {\bibinfo {volume} {78}},\ \bibinfo
		{pages} {032326} (\bibinfo {year} {2008})}\BibitemShut {NoStop}%
	\bibitem [{\citenamefont {Arrazola}\ \emph {et~al.}(2012)\citenamefont
		{Arrazola}, \citenamefont {Gittsovich},\ and\ \citenamefont
		{L\"utkenhaus}}]{Arrazola12}%
	\BibitemOpen
	\bibfield  {author} {\bibinfo {author} {\bibfnamefont {J.~M.}\ \bibnamefont
			{Arrazola}}, \bibinfo {author} {\bibfnamefont {O.}~\bibnamefont
			{Gittsovich}}, \ and\ \bibinfo {author} {\bibfnamefont {N.}~\bibnamefont
			{L\"utkenhaus}},\ }\href {\doibase 10.1103/PhysRevA.85.062327} {\bibfield
		{journal} {\bibinfo  {journal} {Phys. Rev. A}\ }\textbf {\bibinfo {volume}
			{85}},\ \bibinfo {pages} {062327} (\bibinfo {year} {2012})}\BibitemShut
	{NoStop}%
	\bibitem [{\citenamefont {Aghayar}\ \emph {et~al.}(2015)\citenamefont
		{Aghayar}, \citenamefont {Heshmati},\ and\ \citenamefont
		{Jafarizadeh}}]{obwiednie}%
	\BibitemOpen
	\bibfield  {author} {\bibinfo {author} {\bibfnamefont {K.}~\bibnamefont
			{Aghayar}}, \bibinfo {author} {\bibfnamefont {A.}~\bibnamefont {Heshmati}}, \
		and\ \bibinfo {author} {\bibfnamefont {M.~A.}\ \bibnamefont {Jafarizadeh}},\
	}\href@noop {} {\bibfield  {journal} {\bibinfo  {journal} {arXiv:1507.06979
				[quant-ph]}\ } (\bibinfo {year} {2015})}\BibitemShut {NoStop}%
	\bibitem [{\citenamefont {Peres}(1996)}]{Peres96}%
	\BibitemOpen
	\bibfield  {author} {\bibinfo {author} {\bibfnamefont {A.}~\bibnamefont
			{Peres}},\ }\href {\doibase 10.1103/PhysRevLett.77.1413} {\bibfield
		{journal} {\bibinfo  {journal} {Phys. Rev. Lett.}\ }\textbf {\bibinfo
			{volume} {77}},\ \bibinfo {pages} {1413} (\bibinfo {year}
		{1996})}\BibitemShut {NoStop}%
	\bibitem [{\citenamefont {G\"uhne}\ \emph {et~al.}(2003)\citenamefont
		{G\"uhne}, \citenamefont {Hyllus}, \citenamefont {Bruss}, \citenamefont
		{Ekert}, \citenamefont {Lewenstein}, \citenamefont {Macchiavello},\ and\
		\citenamefont {Sanpera}}]{Guhne03}%
	\BibitemOpen
	\bibfield  {author} {\bibinfo {author} {\bibfnamefont {O.}~\bibnamefont
			{G\"uhne}}, \bibinfo {author} {\bibfnamefont {P.}~\bibnamefont {Hyllus}},
		\bibinfo {author} {\bibfnamefont {D.}~\bibnamefont {Bruss}}, \bibinfo
		{author} {\bibfnamefont {A.}~\bibnamefont {Ekert}}, \bibinfo {author}
		{\bibfnamefont {M.}~\bibnamefont {Lewenstein}}, \bibinfo {author}
		{\bibfnamefont {C.}~\bibnamefont {Macchiavello}}, \ and\ \bibinfo {author}
		{\bibfnamefont {A.}~\bibnamefont {Sanpera}},\ }\href {\doibase
		10.1080/09500340308234554} {\bibfield  {journal} {\bibinfo  {journal}
			{J. Mod. Opt.}\ }\textbf {\bibinfo {volume} {50}},\ \bibinfo
		{pages} {1079} (\bibinfo {year} {2003})} \BibitemShut {NoStop}%
	\bibitem [{\citenamefont {A.}(1972)}]{Jamiolkowski72}%
	\BibitemOpen
	\bibfield  {author} {\bibinfo {author} {\bibfnamefont {A.}~\bibnamefont
			{Jamio{\l}kowski}},\ }\href@noop {} {\bibfield  {journal} {\bibinfo  {journal} {Rep. Math.
				Phys.}\ }\textbf {\bibinfo {volume} {3}},\ \bibinfo {pages} {275} (\bibinfo
		{year} {1972})}\BibitemShut {NoStop}%
	\bibitem [{\citenamefont {Choi}(1975)}]{Choi75}%
	\BibitemOpen
	\bibfield  {author} {\bibinfo {author} {\bibfnamefont {M.-D.}\ \bibnamefont
			{Choi}},\ }\href@noop {} {\bibfield  {journal} {\bibinfo  {journal} {Linear
				Algebra and its Applications}\ }\textbf {\bibinfo {volume} {10}},\ \bibinfo
		{pages} {285} (\bibinfo {year} {1975})}\BibitemShut {NoStop}%
	\bibitem [{\citenamefont {Choi}(1982)}]{Choi82}%
	\BibitemOpen
	\bibfield  {author} {\bibinfo {author} {\bibfnamefont {M.-D.}\ \bibnamefont
			{Choi}},\ }\href@noop {} {\bibfield  {journal} {\bibinfo  {journal} {Proc.
				Sympos. Pure Math.}\ }\textbf {\bibinfo {volume} {38}},\ \bibinfo {pages}
		{583} (\bibinfo {year} {1982})}\BibitemShut {NoStop}%
	\bibitem [{\citenamefont {Piani}\ and\ \citenamefont {Mora}(2007)}]{PPTStates}%
	\BibitemOpen
	\bibfield  {author} {\bibinfo {author} {\bibfnamefont {M.}~\bibnamefont
			{Piani}}\ and\ \bibinfo {author} {\bibfnamefont {C.~E.}\ \bibnamefont
			{Mora}},\ }\href {\doibase 10.1103/PhysRevA.75.012305} {\bibfield  {journal}
		{\bibinfo  {journal} {Phys. Rev. A}\ }\textbf {\bibinfo {volume} {75}},\
		\bibinfo {pages} {012305} (\bibinfo {year} {2007})}\BibitemShut {NoStop}%
	\bibitem [{\citenamefont {Horodecki}\ and\ \citenamefont
		{Horodecki}(1996)}]{HH96}%
	\BibitemOpen
	\bibfield  {author} {\bibinfo {author} {\bibfnamefont {R.}~\bibnamefont
			{Horodecki}}\ and\ \bibinfo {author} {\bibfnamefont {M.}~\bibnamefont
			{Horodecki}},\ }\href {\doibase 10.1103/PhysRevA.54.1838} {\bibfield
		{journal} {\bibinfo  {journal} {Phys. Rev. A}\ }\textbf {\bibinfo {volume}
			{54}},\ \bibinfo {pages} {1838} (\bibinfo {year} {1996})}\BibitemShut
	{NoStop}%
	\bibitem [{\citenamefont {Pittenger}\ and\ \citenamefont
		{Rubin}(2000)}]{Pittenger00}%
	\BibitemOpen
	\bibfield  {author} {\bibinfo {author} {\bibfnamefont {A.~O.}\ \bibnamefont
			{Pittenger}}\ and\ \bibinfo {author} {\bibfnamefont {M.~H.}\ \bibnamefont
			{Rubin}},\ }\href {\doibase https://doi.org/10.1016/S0030-4018(00)00612-X}
	{\bibfield  {journal} {\bibinfo  {journal} {Opt. Comm.}\ }\textbf
		{\bibinfo {volume} {179}},\ \bibinfo {pages} {447 } (\bibinfo {year}
		{2000})}\BibitemShut {NoStop}%
	\bibitem [{\citenamefont {Jungnitsch}\ \emph {et~al.}(2011)\citenamefont
		{Jungnitsch}, \citenamefont {Moroder},\ and\ \citenamefont
		{G\"uhne}}]{Moroder11}%
	\BibitemOpen
	\bibfield  {author} {\bibinfo {author} {\bibfnamefont {B.}~\bibnamefont
			{Jungnitsch}}, \bibinfo {author} {\bibfnamefont {T.}~\bibnamefont {Moroder}},
		\ and\ \bibinfo {author} {\bibfnamefont {O.}~\bibnamefont {G\"uhne}},\ }\href
	{\doibase 10.1103/PhysRevLett.106.190502} {\bibfield  {journal} {\bibinfo
			{journal} {Phys. Rev. Lett.}\ }\textbf {\bibinfo {volume} {106}},\ \bibinfo
		{pages} {190502} (\bibinfo {year} {2011})}\BibitemShut {NoStop}%
\end{thebibliography}
\end{document}